\documentclass{llncs}
\usepackage{wrapfig,makeidx}  
\usepackage{mathpartir, amsmath, amssymb}  
\usepackage[all]{xy}
\usepackage{verbatim}
\usepackage{multirow}

\let\llncssubparagraph\subparagraph
\let\subparagraph\paragraph
\usepackage[compact]{titlesec}
 \usepackage[compact]{titlesec}
 \titlespacing{\section}{0pt}{1ex}{1ex}
 \titlespacing{\subsection}{0pt}{1ex}{0ex}
 \titlespacing{\subsubsection}{0pt}{0.5ex}{0ex}
\let\subparagraph\llncssubparagraph

\begin{document}

\title{B\"uchi Types for Infinite Traces and Liveness}
\titlerunning{B\"uchi Types for Infinite Traces and Liveness}
\author{Martin Hofmann\inst{1} \and Wei Chen\inst{2}}
\authorrunning{Martin Hofmann \and Wei Chen}
\institute{LMU Munich
\email{martin.hofmann@ifi.lmu.de}
\and
School of Informatics, University of Edinburgh
}

\maketitle

\begin{abstract}
We develop a new type and effect system based on B\"uchi automata to capture finite and infinite traces produced by programs in a small language which allows non-deterministic choices and infinite recursions. There are two key technical contributions: (a) an abstraction based on equivalence relations defined by the policy B\"uchi automata, the B\"uchi abstraction; (b) a novel type and effect system to correctly capture infinite traces. We show how the B\"uchi abstraction fits into the abstract interpretation framework and show
 soundness and completeness. 
\end{abstract}


\section{Introduction}
A great range of techniques and tools have been developed and studied
for the prediction of program behaviours without actually running the program  \cite{Emerson80,Clarke00,McMillan93,Nielson05,cousot,cousotj}. One of them, originating from type inference in
functional programming languages, is the type and
effect discipline \cite{Luc88}. As a refinement of type systems
in programming languages, types are annotated with information
characterizing dynamic behaviours of programs---{\it
  effects}. As a result, a well-typed program satisfies some properties
regarding its side-effects as well. This type-based technique has been used for all kinds of static analysis of programs, e.g. flow
analysis \cite{Mossin97}, dependency analysis \cite{Abadi99}, resource allocation analysis \cite{Thiemann98}, and
amortised analysis \cite{Hofmann13,Hofmann06}, etc. 
In particular, a type and effect system was developed by Grabowski et al.\ 
\cite{martin} to
verify that a particular programming guideline for secure
web-programming has been adhered to. Generalizing from this, 
one could  model a
programming guideline as a property of traces that a program might
have, where traces are sequences of events that are issued by a
certain instrumentation of the program with special event-issuing
operations. This instrumentation would be part of the formalised
guideline.  A finite state machine would then be used to specify the
set of acceptable traces. 
Most policies involve safety properties  which can be assessed by examining finite portions of
traces. In some cases, however, properties pertaining to liveness and
fairness \cite{Alpern87} can become relevant. For instance, a guideline could be that
calls to appropriate logging functions must be made again and again or
that event (sic!) handlers should not become stuck, e.g.\ in the Java
Swing framework.

This motivated us to investigate the possibility of using type systems
in this situation as well. Our aim is not to offer new algorithms for
deciding certain temporal properties or indeed to compete with the
existing methods which are numerous
\cite{Emerson80,Clarke00,McMillan93,DBLP:conf/concur/BouajjaniEM97},
but to extend the reach of type systems. Our solution goes, however,
beyond a simple reformulation of an existing algorithm; the abstract
domain based B\"uchi automata may well be useful in its on right and
is an original contribution of this work.
 
For the sake of simplicity, we introduce and study a small
language consisting of recursive first-order  procedures 
 and non-deterministic choices. The language explicitly allows 
 infinite recursions. In this language, except
for primitive procedures which have events as arguments, other
procedures have no inputs. We define {\it trace semantics} which
formalize finite and infinite traces generated by programs in this
language. We also remark that in essence our language is the same as
the \emph{pushdown systems} that have been studied in detail by a number of authors 
\cite{DBLP:conf/cav/Walukiewicz96,DBLP:conf/concur/BouajjaniEM97,schwoon}. Similar trace semantics were
also studied by the Cousots \cite{CousotC92} as a specific case of abstract
interpretation.

Once a satisfactory type system for such a simple language has been
found, it can be combined with known techniques \cite{comlan,attapl}
to scale to a type system for a large fragment of Java or similar
languages. Alternatively, one can use our simple language as a target
of a preliminary abstraction step. 

Then, we develop the  B\"uchi type and effect system to capture correctly finite
and infinite traces. Since branching is non-deterministic in our
language, we can even establish a completeness result. Completeness,
of course, will be lost, once we re-introduce data-dependent
branching.

As a demonstration, we extend the B\"uchi type and effect system for this small language to 
a B\"uchi type and effect system for Featherweight Java with field update \cite{comlan}. 

The main technical contribution of this paper is the design of an
abstract domain in the sense of abstract interpretation
\cite{cousot,cousotj} based on B\"uchi automata or rather a mild
extension of those allowing infinite as well as finite words.  The
proofs of soundness and completeness of the type system are based on
clear-cut lattice-theoretic properties of this abstraction.

As in the finitary case, this \emph{B\"{u}chi abstraction} is based on
equivalence relations on finite words generated by the policy
B\"{u}chi automaton. Abstracted effects are no longer sets of such
equivalence classes, but rather sets of pairs of the form $(U,V)$ with
$U$, $V$ classes and representing the infinitary language
$UV^\omega$. While such pairs appear in B\"uchi's original
complementation construction for B\"uchi automata \cite{buc62} and have
subsequently been used by a number of authors
\cite{vardi,martin:c35,jones}, they have never been used in the context of
type systems and abstract interpretation.

\subsection{Related work}
As already mentioned, our language is equivalent to pushdown systems
for which model-checking of temporal properties has been extensively
studied
\cite{DBLP:conf/cav/Walukiewicz96,schwoon,schwoon2,DBLP:journals/tcs/BurkartS99}. Pushdown
systems, on the other hand, are special cases of higher-order recursion
systems introduced by Knapik et
al. \cite{DBLP:conf/fossacs/KnapikNU02} and extensively studied by
Ong and his collaborators,
e.g. \cite{DBLP:conf/tlca/AehligMO05,DBLP:conf/lics/KobayashiO09}.

The latter work \cite{DBLP:conf/lics/KobayashiO09} also casts model
checking into the form of a type system and is thus quite closely
related to our result. More precisely, from an alternating parity
automaton a type system for higher-order recursion schemes is derived
such that a scheme is typable iff its evaluation tree would be
accepted by the automaton. In this way, in particular all mu-calculus
definable properties of the evaluation tree become
expressible. Regarding trace languages as opposed to tree properties
alternating parity automata are equivalent to B\"uchi
automata since both capture the ($\omega-$)regular languages.  Thus for
the trace language of interest here the system from loc.cit.\ is equal in
expressive power to our type system. 

The difference is that Kobayashi and Ong's system has a much more
semantic flavour not unlike the intersection type systems used to
characterise strong normalisation. More concretely, the
well-formedness condition for recursions in that system requires the
solution of a parity game whose size is proportional to the size of
the program (number of function symbols to be precise) which is known
to be equivalent to model checking trees against mu-calculus formulas.

Our type system, on the other hand, deviates from the
standard type systems used in programming and program analysis only
very slightly; instead of the usual recursion rule (which is clearly
unsound in the context of liveness) we use a rule involving a type
variable. No further semantic conditions need to be checked once of course 
the given B\"uchi automaton has been analysed and preprocessed. 

We can also mention that our method and approach are rather
different. While \emph{loc.cit.} uses games and automata we rely on
the recently re-popularised \cite{DBLP:journals/corr/abs-1110-6183,DBLP:conf/concur/AbdullaCCHHMV11} Ramseyian approach to the study of
$\omega$-regular language and automata.

Another recent work on the use of types for properties of infinite
traces is \cite{alan.jeffrey:dblp:c33} which embeds formulas of Linear
Temporal Logic into types in the context of functional reactive
programming. This work, however, relies on an encoding of linear
temporal logic in first-order logic with integers, e.g., one models
``the event $x$ occurs infinitely often'' as a formula like $\forall
i\exists j.x_j$ where $x_j$ refers to that the event $x$ issues at the
time $j$. Dependent types are being used to turn this into a type
system, but questions of inference and decidability are not
considered.

The discussed works 
\cite{DBLP:conf/lics/KobayashiO09,alan.jeffrey:dblp:c33} are---to our knowledge---the only
attempts at extending the range of typing beyond safety properties.
\subsection{Outline}
In the next section we define a simple first-order language with
parameterless recursive procedures and non-deterministic branching
(meant, of course, to abstract ordinary conditionals). An alphabet of
events $\Sigma$ is assumed and for each event $a\in\Sigma$ a primitive
procedure $o(a)$ is available that outputs $a$ and has no effect on
control flow or state. Programs are toplevel mutually recursive
definitions of parameterless procedures comparable to the C-language. Given a program, any expression then admits a set of finite and infinite words over $\Sigma$---the traces of terminating and nonterminating computations of the program. 
We distinguish finite traces stemming from terminating execution from finite traces stemming from nonterminating but ``unproductive'' executions. Thus, for every expression $e$ (relative to a well-formed program) and trace $w\in\Sigma^{\leq\omega} = \Sigma^*\cup\Sigma^\omega$  we define a judgement $e\Downarrow w$ meaning that $e$ admits a terminating execution with trace $w$ (necessarily $w\in\Sigma^*$ then) and another judgement $e\Uparrow w$ meaning that $e$ admits a nonterminating computation with trace $w$. In this case both $w\in\Sigma^*$ and $w\in\Sigma^\omega$ are possible. Formally, 
this is defined by introducing $\checkmark$ events that are repeatedly issued so that any nonterminating computation will have an infinite trace \emph{with} $\checkmark$ events. The official trace semantics ($e\downarrow w$ and $e\uparrow w$) is then defined by discarding these $\checkmark$ events.  We then discuss alternative ways for defining the trace semantics and emphasize that it is merely meant to formalize the intuitively clear notion of event trace occurring during a  computation.

In Section~\ref{sec-e} we then define a type-and-effect system whose
effects are pairs $(U,V)$ with $U\subseteq\Sigma^*$ and $V\subseteq
\Sigma^{\leq\omega}$. Semantically, an expression has effect $(U,V)$
  when $e\downarrow w$ implies $w\in U$ and $e\uparrow w$
  implies $w\in V$. The typing rules are given in
  Figure~\ref{fig-e}. We notice here that for the $U$-part
  (terminating computation) the typing rules are as usual; one
  ``guesses'' a type for a recursively defined procedure and justifies
  it for its body. The rule for the nonterminating ``$V$-part'' is
  different. One assumes a type (and effect) variable for the
  recursive calls, typechecks the body and then takes the greatest
  fixpoint of the resulting type-and-effect equation. 

  With an ordinary recursive typing rule it would be possible to infer
  an effect like $(\emptyset,(a^*b)^\omega)$ (``infinitely often
  $b$'') for the program $m() = o(a);m()$ which is unsound.  We then
  establish soundness (Theorem~\ref{thm-a}) and completeness
  (Theorem~\ref{thm-ca}) for this type-and-effect system. In
  particular, this shows that the proposed handling of recursive
  definitions does indeed work.

 The type system is at this level, however, of limited use since the
 effects are infinitary objects. Therefore, in Section~\ref{sec-q} we
 introduce an abstraction of this type-and-effect system where effects are taken
 from a fixed finite set. This finite set is calculated from an a priori
 given B\"uchi automaton and effects still denote pairs of
 finite and possibly infinite languages, but no longer is any such
 pair denotable. 

Our main result Theorem~\ref{thm-cb} then asserts
 that if the set of all traces of an expression is accepted by the
 given B\"uchi automaton then this is provable in the abstracted type-and-effect
 system. So, no precision is lost in this sense. Of course, we also
 have an accompanying soundness theorem (Theorem~\ref{thm-c}) for the
 abstract type-and-effect system.

 These results can be modularly deduced
 from soundness and completeness for the infinitary type-and-effect system
 (Thms~\ref{thm-a} and \ref{thm-ca}) with the help of lattice-theoretic properties of the abstraction that are established in Section~\ref{ba}. 
In particular, we have  a Galois connection between the lattice of all
 languages and the lattice of language denotations in the abstract type
 system and all operations needed in the typing rules, in particular least and greatest fixpoints can be correctly rendered on the level of the abstractions. 

 The crucial building block is the ability to compute abstractions of
 greatest fixpoints needed for recursive definitions entirely on the
 level of the abstracted types. This requires the combination of a
 combinatorial lemma (Lemma \ref{krit}) with known covering properties
 (Lemma~\ref{lem-d}) of the abstractions which follow from Ramsey's
 theorem. Section~\ref{xx} and Section~\ref{ba} contain these
 lattice-theoretic results.  We consider the discovery of this
 abstract lattice obtained from a B\"uchi automaton an important
 result of independent interest.

 Section~\ref{bet} then contains the actual definition of the
 abstracted type-and-effect system and its soundness and completeness theorems
 which, as already mentioned, then are direct consequences of earlier
 results. Section~\ref{compl} discusses automatic type inference and its complexity. 

An Appendix contains several worked out examples that did not fit into
the main text and omitted proofs. We also sketch there, as a
demonstration, a combination of an existing region-based type and
effect system for Featherweight Java with field update \cite{comlan}
with B\"uchi types.

\section{Trace Semantics}
%
The syntax of expressions is given by 
$
e\ ::=\ o(a)\ |\ f\ |\ e_1\ ;\ e_2\ |\ e_1\ ?\ e_2
$
where $o(a)$ is the only primitive procedure which generates an event $a$ taken from a fixed alphabet $\Sigma$ of events and $f$ ranges over procedures defined by expressions. Parentheses are used to eliminate ambiguity. We assume that the operator $;$ is right-associative and has higher priority than the operator $?$. As an example, we can define procedures $f$ and $g$ as:
$f = o(b)\ ?\ o(a)\ ;\ g$ and
$g = f\ ;\ g\ ;\ (o(b)\ ?\ o(a))$.
Formally, thus a \emph{program} consists of a finite set of procedure identifiers $\mathcal{F}$ and for each $f\in\mathcal{F}$ an expression $e_f$ defining $f$ where calls to procedures from $\mathcal{F}$ are allowed and in particular, $f$ may occur recursively in $e_f$.

From now on, we fix such a program
$\mathcal{P}=(\mathcal{F},(e_f)_{f\in\mathcal{F}})$ and call an
expression $e$ \emph{well-formed} if it uses calls to procedures from
$\mathcal{F}$ only.

Since the operator $?$ is non-deterministic and non-primitive
procedures have no arguments, stacks and heaps are not needed at this
level of abstraction. 

Let $\Sigma^{\leq \omega}$ be the set of all finite and infinite
sequences generated from the set $\Sigma$ of primitive events. We call
an element $w$ in $\Sigma^{\leq \omega}$ a {\it trace}. Given
traces $w$ and $u$, we define the concatenation $w\cdot u$ as:
$wu$ if $w \in \Sigma^*$ and 
$w$ if $w \in \Sigma^\omega$
where $\Sigma^*$ and $\Sigma^\omega$ are respectively sets of all
finite and infinite sequences over $\Sigma$. So,
$\Sigma^{\leq\omega} = \Sigma^*\cup\Sigma^\omega$ and $\Sigma^* =
\Sigma^+\cup\{\epsilon\}$. As usual, we may write $wu$ instead of $w\cdot u$. 
We are concerned with finite prefixes of the trace generated by
a given  expression. We call them {\it observed
  traces}. Notice that all observed traces are in $\Sigma^*$. Let
$e_f$ be the definition (a well-formed expression) of $f$. The
observed trace semantics is given in Figure \ref{fig-a}.
{\small\begin{figure}[h]
{\scriptsize\[
\inferrule*
{ }
{o(a) \Downarrow a}
\quad
\inferrule*
{ }
{o(a) \Uparrow a}
\quad
\inferrule*
{ }
{e \Uparrow \epsilon}
\quad
\inferrule*
{e_f \Downarrow w}
{f \Downarrow w}
\quad
\inferrule*
{e_f \Uparrow w}
{f \Uparrow w}
\inferrule*
{e_1 \Downarrow w\qquad e_2 \Downarrow u}
{e_1\,;\, e_2 \Downarrow w \cdot u}
\quad
\inferrule*
{e_1 \Downarrow w\qquad e_2 \Uparrow u}
{e_1\,;\,e_2 \Uparrow w \cdot u}
\]\[
\quad
\inferrule*
{e_1 \Uparrow w}
{e_1\,;\,e_2 \Uparrow w}
\inferrule*
{e_1 \Downarrow w}
{e_1\,?\,e_2 \Downarrow w}
\quad
\inferrule*
{e_2 \Downarrow w}
{e_1\,?\,e_2 \Downarrow w}
\quad
\inferrule*
{e_1 \Uparrow w}
{e_1\,?\,e_2 \Uparrow w}
\quad
\inferrule*
{e_2 \Uparrow w}
{e_1\,?\,e_2 \Uparrow w}
\]}
\caption{The Observed Trace Semantics}
\label{fig-a}
\end{figure}
}
We write $e\Downarrow w$ to mean that the finite trace generated by $e$ is $w$. In particular, $e$ terminates. We write $e\Uparrow w$ to mean that $w$ is a finite prefix of the trace generated by $e$. Let the notation $u \preccurlyeq w$ denote that $u$ is a finite prefix of $w$. We have: 
if $e \Downarrow w$ or $e\Uparrow w$, then for all $u\preccurlyeq w$, $e \Uparrow u$.
We now turn to define infinite traces of non-terminating programs. Unfortunately, the observed trace semantics does not contain enough information for this. Let us consider the following definitions:
$f = o(a)$, 
$g = (o(a)\ ;\ h)\ ?\ o(a)$, and
$h = h$.
Notice that the observed traces of $f$ and $g$ are exactly the
same. 
However, the procedure $g$ has a path leading to an unproductive infinite recursion $h$ while $f$ is non-recursive. 
In order to fix this problem, let us introduce the extended set $\Sigma\uplus\{\checkmark\}$ of events and use $(\Sigma\uplus\{\checkmark\})^{\leq\omega}$ for the set of all {\it extended traces}. The \emph{observed extended trace semantics} is same as the observed trace semantics except for the rule for function application in which a $\checkmark$-event is automatically generated. 
That is,
$
\frac
{e_f \Uparrow w}
{f \Uparrow \checkmark\cdot w}
$. 
The specific symbol $\checkmark$ is added to the beginning of trace $w$ of $e_f$. By doing this, unproductive infinite recursions can be distinguished from productive cases by observed extended traces $\checkmark^*$.
%

For all observed extended traces $w$, let $\theta(w)$ denote the trace obtained from $w$ by removing all $\checkmark$s. 
Based on the observed extended trace semantics, we define trace semantics as follows.
\begin{definition}[Trace Semantics]
For all expressions $e$ and extended traces $w$ in $(\Sigma\uplus\{\checkmark\})^{\leq \omega}$,
\begin{eqnarray*}
& & e \downarrow w\ \equiv\ \exists w'\in(\Sigma\uplus \{\checkmark\})^*.\,e\Downarrow w' \wedge w = \theta(w');\\
& & e \uparrow w\ \equiv\ \exists w'\in(\Sigma\uplus \{\checkmark\})^\omega.\,(\forall u\preccurlyeq w'\,.\, e\Uparrow u)
\,\wedge\, w = \theta(w').
\end{eqnarray*}
We say $w$ is a trace of $e$ if $e\downarrow w$ or $e\uparrow w$.
\end{definition}
Notice that if $e\downarrow w$, then $w$ is in $\Sigma^*$ and all executions of $e$ terminate. If $e\uparrow w$, then $w$ is in $\Sigma^{\leq\omega}$ and all executions of $e$ do not terminate.
In our definition of trace semantics, the symbol $\checkmark$ is introduced to distinguish finite traces generated by terminating programs and non-terminating programs. When the trace semantics is well-defined, we remove all $\checkmark$s.  

We remark that this way of defining the semantics is one of several
possibilities; alternatives would consist of using a small step
operational semantics or a coinductive definition. For instance, Cousot et al 
\cite{CousotC92} define a generalization of structured operational
semantics (G$^\infty$SOS), is used to describe the finite and
infinite executions of programs. At the end of the day we need to
define the two judgements $e\downarrow w$ meaning that $e$
terminates with trace $w$ so, necessarily $w\in\Sigma^*$ and
$e\uparrow w$ meaning that $e$ does not terminate (runs forever)
and its trace is $w$. In this case, $w$ may either be an
infinite word ($w\in\Sigma^\omega$) or a finite word ($w\in
\Sigma^*$) in which case $e$'s evaluation gets stuck in an infinite
loop but $e$ does not output events during this loop.

An important fine point is that at our level of abstraction programs have a finite store which means that by K\"onig's lemma ``arbitrarily long'' and ``infinitely long'' coincide. In a language allowing the  nondeterministic selection of integers we could write a program that admits traces (outputting $a$s) of any finite length but not having an infinite trace. Then, our trace semantics would erroneously ascribe the trace $a^\omega$ to such a program. But, fortunately, in our situation this does not occur. As a result, for some language extensions, one may need to consider more complicated formal definitions of trace semantics. This would, however, have no influence on the type system we define and only very little influence on correctness proofs.

\section{Type and Effect System}
\label{sec-e}
In this section, we develop a type and effect system that captures the
set of finite and infinite traces of a program.  We also prove that
this system is sound and complete. This system uses arbitrary
languages for effect annotations and as such is not yet suitable for
practical use let alone automatic inference. Later, in Section~\ref{sec-q}
we define a finitary abstraction of this system which still allows one
to check soundly and completely whether the traces of a given program
are accepted by a fixed B\"uchi automaton.

\begin{definition}[Effect]
Let $U$ be a subset of $\Sigma^*$ and $V$ be a subset of $\Sigma^{\leq\omega}$. An effect of a given expression $e$ is a pair $(U, V)$ satisfying:
(a) if $e\downarrow w$, then $w$ is in $U$;
(b) if $e\uparrow w$, then $w$ is in $V$. 
We use the notation $e\ \&\ (U,V)$ to denote that $(U,V)$ is an effect of $e$.
\end{definition}
Let $\mathfrak{X}$ be a set of variables. Let $V(\mathfrak{X})$ range over expressions of the form:
$\bigcup_{X\in\mathfrak{X}} (A_X\cdot X) \cup B$
with $A_X \subseteq \Sigma^*$ and $B\subseteq \Sigma^{\leq\omega}$. We abbreviate $\mathfrak{X}\setminus\{X\}$ by $\mathfrak{X}-X$ and thus use the notation $V(\mathfrak{X}-X)$ to denote expressions of the form:
$\bigcup_{Y\in\mathfrak{X}-\{X\}} (A_Y\cdot Y) \cup B$. 
We use the symbol $X$ itself for the expression where $B=\emptyset$, $A_X = \{\epsilon\}$, and $A_Y = \emptyset$ for all $Y$ in $\mathfrak{X}-X$. We define the following operations on these expressions:
$A \cdot V(\mathfrak{X}) = \bigcup_{X\in\mathfrak{X}} ((A\cdot A_X)\cdot X) \cup (A\cdot B)$ and
$V(\mathfrak{X}) \cup V'(\mathfrak{X})= \bigcup_{X\in\mathfrak{X}} ((A_X \cup A_X')\cdot X) \cup (B\cup B')$
where $A$ is a subset of $\Sigma^*$. Given an assignment function $\eta : \mathfrak{X} \to \mathcal{P}(\Sigma^{\leq\omega})$ that assigns a set of traces to each variable $X$ in $\mathfrak{X}$, we obtain for each expression $V(\mathfrak{X})$ a language $V(\eta)\subseteq\Sigma^{\leq\omega}$ by substituting $\eta(X)$ for each variable $X$.

We define $A^\omega$ as the set of all words of the form 
$w= w_0w_1w_2\cdots w_i\cdots$ where $w_i\in A$. Note that $A^\omega\subseteq\Sigma^{\leq\omega}$. 
Let $\Delta$ be an {\it environment} that is a set of expressions of the form:
$f\ \&\ (U, X)$
with $f$ a non-primitive procedure, $U$ a subset of $\Sigma^*$, and $X$ a variable in $\mathfrak{X}$ such that if $f\ \&\ (U,X)$ and $g\ \&\ (V,Y)$ both occur in $\Delta$ then $X\neq Y$. With the above definitions, we define the {\it type-and-effect system} in Figure \ref{fig-e}. 
\begin{figure}[h]
{\scriptsize\[
\inferrule*
{ }
{\Delta\vdash o(a)\ \&\ (\{a\}, \emptyset)}\quad 
\inferrule*
{
{\Delta\vdash e_1\ \&\ (U_1, V_1(\mathfrak{X}))}\qquad
{\Delta\vdash e_2\ \&\ (U_2, V_2(\mathfrak{X}))}
}
{\Delta\vdash e_1\,;\,e_2\ \&\ (U_1\cdot U_2, V_1(\mathfrak{X}) \cup U_1\cdot V_2(\mathfrak{X}))}
\]
\[
\inferrule*
{
{\Delta\vdash e_1\ \&\ (U_1, V_1(\mathfrak{X}))}\qquad
{\Delta\vdash e_2\ \&\ (U_2, V_2(\mathfrak{X}))}
}
{\Delta\vdash e_1\,?\,e_2\ \&\ (U_1\cup U_2, V_1(\mathfrak{X})\cup V_2(\mathfrak{X}))}
\quad \inferrule*
{
}
{\Delta, f\ \&\ (U,X)\vdash f\ \&\ (U,X)}
\]
\[
\inferrule*
{
{\Delta, f\ \&\ (U,X)\vdash e_f\ \&\ (U,A\cdot X \cup V(\mathfrak{X}-X))}
}
{\Delta\vdash f\ \&\ (U,A^*\cdot V(\mathfrak{X}-X)\cup A^\omega)}
\]}
\caption{The Type and Effect System}
\label{fig-e}
\end{figure}
An environment $\Delta$ is justified\, if  for all $f\ \&\ (U, X)$ in $\Delta$ one has
$\Delta\vdash e_f\ \&\ (U, A\cdot X\cup V(\mathfrak{X}-X))$ for some $A$, $V(\mathfrak{X}-X)$. 
A justified environment can be extended as follows: 
\begin{lemma}
\label{lem-a}
Given a justified environment $\Delta$ such that 
$\Delta, f\ \&\ (U,X) \vdash e_f\ \&\ (U,A\cdot X\cup V(\mathfrak{X}-X))$,
then the extended environment $\Delta, f\ \&\ (U,X)$ is also justified
\end{lemma}
An assignment function $\eta$ \emph{satifies} an environment $\Delta$  if whenever $f\ \&\ (U, X)$ in $\Delta$ and 
$f\uparrow w$ then  $w\in \eta(X)$.
Let us use the notation $\eta \models \Delta$ to denote that the environment $\Delta$ is justified and that the assignment function $\eta$ satisfies $\Delta$. 
\begin{lemma}
\label{lem-b}
Given an environment $\Delta$ and an assignment function $\eta$ satisfying that $\eta \models \Delta$, let $\eta'$ be an extension $\eta[X\mapsto V]$ of $\eta$ such that $f\uparrow w$ implies $w\in V$ for all traces $w$. If we have the derivation: 
$\Delta, f\ \&\ (U,X) \vdash e_f\ \&\ (U, A\cdot X\cup V(\mathfrak{X}-X))$,
then $\eta' \models \Delta, f\ \&\ (U,X)$.
\end{lemma}

\begin{theorem}[Soundness]
\label{thm-a}
Given an environment $\Delta$ and an assignment function $\eta$ satisfying that $\eta \models \Delta$, for all derivations: 
$\Delta\vdash e\ \&\ (U, V(\mathfrak{X}))$
of an expression $e$, we have:
$e\downarrow w$ implies $w \in U$ and
$e\uparrow w$ implies $w\in V(\eta)$.
\end{theorem}
\begin{proof}
The only interesting case is that for the last rule in Figure \ref{fig-e} which relies on Lemma \ref{lem-b}. For more details, see Appendix \ref{a-p}.
\end{proof}
\begin{corollary}
\label{cor-a}
For all derivations
$\vdash e\ \&\ (U,V)$
of an expression $e$, we have:
$e \downarrow w$ implies $w\in U$ and
$e \uparrow w$ implies $w\in V$.
\end{corollary}
Fix for each non-primitive procedure $f\in\mathcal{F}$ a unique variable $X_f$. If
$\vec A = (A_f)_f$ is a family of languages with
$A_f\subseteq\Sigma^*$ define the corresponding environment
$\Delta(\vec A)$ as to contain the bindings $f\ \&\ (A_f,X_f)$. For each
function body $e_f$ we can now derive using the rules except the last
one a unique typing $\Delta(\vec A)\vdash e_f\ \&\ \dots$. The passage
from $\vec A$ to $\vec C=(C_f)_f$ defines a monotone operator $\Phi$
on the lattice $\mathcal{P}(\Sigma)^{\mathcal{F}}$. If $\vec B$ is
the least fixpoint of this operator then $\Delta(\vec B)$ is justified
and we get the judgements $\Delta(\vec B)\vdash e_f\ \&\ (B_f,
V(\mathfrak{X}))$. Successive application of the last rule then gives
judgements $\vdash f\ \&\ (U_f,V_f)$ and a direct induction shows that in
fact $U_f=\{w\mid f\downarrow w\}$ and $V_f=\{w\mid f\uparrow w\}$. We
have thus shown:
\begin{theorem}[Completeness]
\label{thm-ca}
The judgements $\vdash f : (\{w\mid f\downarrow w\},\{w\mid f\uparrow w\})$ are derivable for each $f$. 
\end{theorem}
We have kept the proof of this theorem in the running text since the monotone operator $\Phi$ is still needed later. 
\section{B\"{u}chi Type and Effect System}
\label{sec-q}
Based on equivalence relations on finite words defined by the policy B\"{u}chi automata, we introduce an abstraction of languages of finite and infinite words: the B\"{u}chi abstraction. We place this abstraction into the framework of abstract interpretation and show that crucial operations, namely concatenation, least fixpoint, and infinite iteration ($(-)^\omega$) can be computed on the level of the abstraction. We also show that the abstraction does not lose any information as far as acceptance by the fixed policy automaton is concerned. This then allows us to replace the infinitary effects in the previous type system by their finite abstraction and thus to obtain a type-and-effect system which is decidable with low complexity (in the program size) and yet complete. The soundness and completeness of this system follow directly from lattice-theoretic properties of this B\"{u}chi abstraction (Lemma~\ref{lem-d} and Theorem~\ref{prese}). 
\subsection{Extended B\"{u}chi Automata}\label{xx}
Given an expression $e$, our goal is to verify that the set $T(e)$ of all traces generated by $e$ satisfies some property. We use a mild extension of the standard B\"uchi Automata which we call \emph{extended B\"uchi automata}:
\begin{definition}
An extended B\"uchi Automaton is a quadruple
$
\mathfrak{A} = (Q, \Sigma, \delta, q_0, F)
$
where $Q$ is a finite set of states, $\Sigma$ is an alphabet;
hereafter always required to be equal to the fixed alphabet of events;
$\delta : Q\times \Sigma\to \mathcal{P}(Q)$ the transition function, the
initial state $q_0\in Q$, and the set $F \subseteq Q$ of final
states. The language $L(\mathfrak{A})$ of $\mathfrak{A}$ is defined
as: the set of all finite words by which a final state can be reached
from the initial state and all infinite words for which there is a
path which starts from the initial state and goes through final states
infinitely often. Thus, $L(\mathfrak{A})$ is the union of
$\mathfrak{A}$'s language when understood as a traditional NFA and its
language when understood as a traditional B\"uchi automaton.
\end{definition}


Following B\"uchi's original works we use equivalence relations
defined by extended B\"{u}chi automata themselves to obtain finite
representations of $U$ and $V$. We write $q \stackrel{w}{\leadsto}
q'$ to mean that the state $q'$ is reachable from the state $q$ by
using the finite word $w$. Let $q\stackrel{w}{\leadsto}_F q'$
denote that by using the finite word $w$, the state $q'$ can be
reached from the state $q$ in such a way that a final state is visited
on the way. In particular, $q\stackrel{w}{\leadsto}q'$ with $q\in
F$ or $q'\in F$ implies $q\stackrel{w}{\leadsto}_F q'$. Formally,
we have $q\stackrel{w}{\leadsto}_F q'$ iff there exists $q''\in F$
and $u,v$ such that $w=uv$ and $q
\stackrel{u}{\leadsto} q''$ and $q'' \stackrel{v}{\leadsto}
q'$.

For nonempty words $w,u\in\Sigma^+$ we define
\begin{eqnarray*}
w \sim u\ \equiv\ \forall p,q\in Q\ \, .\ \, p\stackrel{w}{\leadsto} q \Leftrightarrow p\stackrel{u}{\leadsto} q\ \wedge\ p\stackrel{w}{\leadsto}_F q \Leftrightarrow p\stackrel{u}{\leadsto}_F q\ .
\end{eqnarray*}
We write $[w]$ for the equivalence class of $w\in\Sigma^+$ and additionally 
let $[\epsilon]$ stand for $\{\epsilon\}$. We write $\mathcal{Q}$ for $\Sigma^+/{\sim}\uplus \{[\epsilon]\}$. Thus  $\mathcal{Q}$ comprises the $\sim$-equivalence classes and a special class for the empty word. 

We notice that if $w\sim u$ and $w'\sim u'$ then $ww'\sim uu'$. As a result concatenation is well-defined on equivalence classes, thus $\Sigma^+/\sim$ becomes a semigroup and $\mathcal{Q}$ a monoid. The following Lemma is a straightforward consequence of standard results about B\"uchi automata \cite{Tho-languages-al}. 
\begin{lemma}
\label{lem-d}
Fix an extended B\"{u}chi automaton $\mathfrak{A} = (Q, \Sigma, \delta, q_0, F)$, 
\begin{itemize}
\item[(a)] $\mathcal{Q}$ is finite and  its elements  are regular languages;
\item[(b)] for all classes $C$ in $\mathcal{Q}$,
$C \cap  L(\mathfrak{A}) \neq \emptyset$ implies
$C \subseteq  L(\mathfrak{A})$;  
\item[(c)] for all classes $C$ and $D$ in $\mathcal{Q}$,
$CD^\omega \cap  L(\mathfrak{A}) \neq \emptyset$ implies 
$CD^\omega \subseteq  L(\mathfrak{A})$;
\item[(d)] for every word $w\in\Sigma^{\leq \omega}$ there exist classes $C,D\in\mathcal{Q}$ so that $w\in CD^\omega$ and $CD=C$ and $DD=D$.
\end{itemize}
\end{lemma}
The sets $CD^\omega$ (with $CD=C$ and $DD=D$) thus behave almost like classes themselves, but an important difference is that they may nontrivially overlap. If $CD^\omega\cap UV^\omega\neq \emptyset$ then in general one cannot conclude $CD^\omega=UV^\omega$. We also  remark that Ramsey's theorem is used in the proof of (d). 
\subsection{B\"{u}chi Abstraction}\label{ba}
Lemma~\ref{lem-d} shows that given an extended B\"{u}chi automaton $\mathfrak{A}$, without affecting property checking, we can use sets of classes in 
$\mathcal{Q}$ to represent languages over $\Sigma^*$ 
and sets of pairs of classes $(C,D)$ such that $CD=C$ and $DD=D$ to represent languages over $\Sigma^{\leq\omega}$. Let us write $\mathcal{C}:=\{(C,D)\mid C,D\in\mathcal{Q} \wedge CD=C, DD=D\}$. Let us define the pre-abstraction function 
$\mathfrak{f} : \mathcal{P}(\Sigma^{\leq\omega})\to \mathcal{P}(\mathcal{C})$
and the pre-concretization function
$\mathfrak{g} : \mathcal{P}(\mathcal{C})\to\mathcal{P}(\Sigma^{\leq\omega})$
as:
$\mathfrak{f}(V) = \{(C,D)\in\mathcal{C}\mid CD^\omega\cap V\neq\emptyset\}$
and
$\mathfrak{g}(\mathcal{V}) = \bigcup_{(C,D)\in\mathcal{V}}CD^\omega$ respectively.
A set $\mathcal{V}\subseteq \mathcal{C}$ is
\emph{closed}, if
$\mathfrak{f}(\mathfrak{g}(\mathcal{V}))=\mathcal{V}$. Explicitly,
$\mathcal{V}$ is closed if whenever $UV^\omega\cap
CD^\omega\neq\emptyset$ for some $(C,D)\in\mathcal{V}$ and $(U,V)\in\mathcal{C}$  then already $(U,V)\in\mathcal{V}$.

Clearly, for every set $\mathcal{V}\subseteq
\mathcal{C}$ there is a least closed superset and it
is given by applying the {\it closure} function:
$\mathfrak{c} : \mathcal{P}(\mathcal{C})\to\mathcal{P}(\mathcal{C})$ which is defined as:
$\mathfrak{c}(\mathcal{V}) = \bigcup_{n=1}^\infty (\mathfrak{f}\circ \mathfrak{g})^n(\mathcal{V})$.
We write $\mathcal{M}_{\leq\omega} = \{\mathfrak{c}(\mathfrak{f}(V))\mid V\subseteq\Sigma^{\leq\omega}\}$
for the set such closed subsets.
The elements of $\mathcal{M}_{\leq\omega}$ will serve as abstractions of languages over $\Sigma^{\leq\omega}$. 

We also define explicitly 
$\mathcal{M}_* = \mathcal{P}(\mathcal{Q})$
to represent languages over $\Sigma^*$ but note that via the embedding 
$C\mapsto (C,\{\epsilon\})$ we could identify $\mathcal{M}_*$ with a subset of $\mathcal{M}_{\leq\omega}$. 
\begin{lemma}
\label{lem-in}
Both $\mathcal{M}_*$ and $\mathcal{M}_{\leq\omega}$ are complete lattices with respect to inclusion. 
\end{lemma}
From now on we call $\mathcal{P}(\Sigma^*)$ and $\mathcal{P}(\Sigma^{\leq\omega})$ the \emph{concrete domains} and $\mathcal{M}_*$ and $\mathcal{M}_{\leq\omega}$ the \emph{abstract domains}. We introduce the following \emph{abstraction functions}:
$\alpha_* : \mathcal{P}(\Sigma^*)\to\mathcal{M}_{*}$ and
$\alpha_{\leq\omega} : \mathcal{P}(\Sigma^{\leq\omega})\to\mathcal{M}_{\leq\omega}$,
which are respectively defined as: 
$\alpha_*(U) = \{C\in\mathcal{Q}\mid C\cap U\neq \emptyset\}$
and
$\alpha_{\leq\omega}(V) = \mathfrak{c}(\mathfrak{f}(V))$.
We also introduce the following \emph{concretization functions}: 
$\gamma_* : \mathcal{M}_{*}\to \mathcal{P}(\Sigma^*)$
and
$\gamma_{\leq\omega} : \mathcal{M}_{\leq\omega}\to \mathcal{P}(\Sigma^{\leq\omega})$,
which are respectively defined as:
$\gamma_*(\mathcal{U}) = \bigcup_{C\in\mathcal{U}}C$
and
$\gamma_{\leq\omega}(\mathcal{V}) = \mathfrak{g}(\mathcal{V})$.
\begin{lemma}
\label{lem-galois}
The abstraction and concretization functions are monotone and form Galois connections, that is: 
$\alpha_*(U)\subseteq \mathcal{U}$ iff $U\subseteq\gamma_*(\mathcal{U})$, and 
$\alpha_{\leq\omega}(V)\subseteq \mathcal{V}$ iff $V\subseteq\gamma_{\leq\omega}(\mathcal{V})$.
Moreover, $\alpha_{*}(\gamma_{*}(\mathcal{U})) = \mathcal{U}$ and 
$\alpha_{\leq\omega}(\gamma_{\leq\omega}(\mathcal{V})) = \mathcal{V}$ so we have in fact a Galois \emph{injection}. 
Furthermore, both abstraction and concretization functions preserve
unions, least and greatest elements.  The concretization functions
also preserve intersections.
\end{lemma}
The next lemma shows that the abstraction is sufficiently fine for our purposes.
It is a direct consequence of the Galois connection and the fact that $L(\mathfrak{A})$ itself is closed which in turn is direct from Lemma~\ref{lem-d} 
(b) and (c). 
\begin{lemma}\label{compless}
Let $p\in\{*,\leq\omega\}$.
If $L\subseteq \Sigma^p$ then 
$\gamma_p(\alpha_p(L))\subseteq L(\mathfrak{A})$ iff $L\subseteq L(\mathfrak{A})$.
\end{lemma}
We now turn to define some new operators for the abstract domains. We have a concatenation operation on $\mathcal{M}_*$ given pointwise, i.e.\ for $\mathcal{U},\mathcal{U}'\in\mathcal{M}_*$, we define $\mathcal{U}\cdot\mathcal{U}'=\{UU'\mid U\in\mathcal{U}, U'\in\mathcal{U}'\}$. We also define a concatenation  $\cdot : \mathcal{M}_*\times\mathcal{M}_{\leq\omega}\to\mathcal{M}_{\leq\omega}$ as follows: 
$\mathcal{U}\cdot \mathcal{V} = \{(AC,D)\mid A\in\mathcal{U} \wedge (C,D)\in\mathcal{V}\}$. Note that $ACD=AC$. 
\begin{lemma}
\label{lem-con}
If $\mathcal{U}\in\mathcal{M}_*$ and  $\mathcal{V}\in\mathcal{M}_{\leq\omega}$ then $\mathcal{U}\cdot\mathcal{V}\in \mathcal{M}_{\leq\omega}$. 
\end{lemma}
\begin{theorem}\label{prese}
The preservation properties listed in Table~\ref{tab-a} are valid.
\end{theorem} 
  Most of these properties are direct and folklore or have been
  asserted earlier. In the clause about least fixpoints denoted by
  $\textit{lfp}$ the operators $\Phi$ and $F$ are supposed to be
  monotone operators on $\mathcal{P}(\Sigma^*)^n$ and
  $\mathcal{M}_*^n$ for some $n$. The clause is then validated by
  straightforward application of lattice-theoretic principles. Only
  preservation of $(-)^\omega$ is a nontrivial and original result; it
  requires the following lemma.
\begin{table} 
{\scriptsize
\begin{center}
\begin{tabular}{ll}
$\alpha_* : \mathcal{P}(\Sigma^*)\to\mathcal{M}_{*}$ &
$\alpha_{\leq\omega} : \mathcal{P}(\Sigma^{\leq\omega})\to\mathcal{M}_{\leq\omega}$\\
$\alpha_*(\top) = \top$ &
$\alpha_{\leq\omega}(\top) = \top$\\
$\alpha_*(\bot) = \bot$ &
$\alpha_{\leq\omega}(\bot) = \bot$\\ 
$\alpha_*(U\cup U') = \alpha_*(U)\! \cup\! \alpha_*(U')$ &
$\alpha_{\leq\omega}(V\cup V') = \alpha_{\leq\omega}(V)\! \cup\! \alpha_{\leq\omega}(V')$\\ 
$\alpha_*(U\cdot U') = \alpha_*(U) \cdot \alpha_*(U')$ &
$\alpha_{\leq\omega}(U\cdot V) = \alpha_{*}(U) \cdot \alpha_{\leq\omega}(V)$\\ 
$\alpha_*\circ \Phi = F\circ\alpha_*\Longrightarrow \alpha_*(\textit{lfp}(\Phi))=\textit{lfp}(F)\qquad$ & 
$\alpha_{\leq\omega}(U^\omega) = \alpha_{*}(U)^\omega$ \\ 
\end{tabular}
\end{center}
}
\caption{Properties of the B\"{u}chi Abstraction}
\label{tab-a}
\end{table} 

\begin{lemma}\label{krit}
Let $(L_i)_{i\in I}$ be a family of classes (from $\mathcal{Q}$)
and put $P=\prod_{i\in I} L_i \subseteq \Sigma^{\leq\omega}$, i.e., $P$ comprises finite or infinite words of the form $w_1w_2w_3\dots$ where $w_i\in L_i$ for $i\geq 1$. There exist classes $U,V\in\mathcal{Q}$ where $UV=U, VV=V$ such that $P \subseteq UV^\omega$. 
\end{lemma} 
\begin{proof}
Let $w \in P$ and write
$w=w_1w_2w_3\cdots w_i\cdots$ where $w_i\in L_i$. If $w$ is a finite word then there exists $n$ such that $w_i=\epsilon$ (and $L_i=[\epsilon]$) 
for $i\geq n$ and we can choose $U=L_1\dots L_{n-1}$ and $V=[\epsilon]$. Otherwise,  use Ramsey's theorem as in the proof of Lemma~\ref{lem-d} to obtain a sequence of indices $i_1 < i_2
  < i_3 < i_4 < \dots$ and classes $U,V$ where $V\neq[\epsilon]$ and $VV=V$, $UV=U$
  such that
$w_1w_2\dots w_{i_1}\in U$,
$w_{i_1+1}\dots w_{i_2}\in V$ and
$w_{i_2+1}\dots w_{i_3}\in V$ and so on.
It follows that $U=L_1L_2\dots L_{i_1}$ and $V=L_{i_k+1}\dots L_{i_{k+1}}$ for $k\geq 1$ and thus $P\subseteq UV^\omega$ as required. 
\end{proof}

\begin{proof}[of Theorem~\ref{prese}]
It only remains to prove preservation of $(-)^\omega$. So, fix $L\subseteq \Sigma^*$. 
We want to show that $\alpha_{\leq\omega}(L^\omega)=\alpha_{*}(L)^\omega$. Note that $\alpha_{*}(L)^\omega=\alpha_{\leq\omega}(\gamma_{*}(\alpha_{*}(L))^\omega)$. The direction ``$\subseteq$'' is obvious from monotonicity; towards proving ``$\supseteq$'' assume $(U,V)\in\alpha_{\leq\omega}(\gamma_*(\alpha_*(L))^\omega)$. Since $\alpha_{\leq\omega}(L^\omega)$ is closed, we may without loss of generality assume that $UV^\omega\cap \gamma_*(\alpha_*(L))^\omega\neq\emptyset$. Pick $w\in UV^\omega\cap \gamma_*(\alpha_*(L))^\omega$ and decompose $w=w_1w_2\dots$ where $w_i\in \gamma_*(\alpha_*(L))$. Define $L_i:=[w_i]$ and apply Lemma~\ref{krit} to obtain $U',V'$ with 
 $\prod_i L_i\subseteq U'{V'}^\omega$. Note that, since $w\in P$, we have $UV^\omega\cap U'{V'}^\omega\neq\emptyset$. 

Now, since $w_i\in\gamma_*(\alpha_*(L))$, by the definition of $\alpha_*$, we must have that $L_i\cap L\neq\emptyset$. Choose $w_i'\in L_i\cap L$. The word $w_1'w_2'\dots$ is then contained in $L^\omega\cap U'{V'}^\omega$, so $(U',V')\in\alpha_{\leq\omega}(L^\omega)$ and, finally, $(U,V)\in \alpha_{\leq\omega}(L^\omega)$ since $\alpha_{\leq\omega}(L^\omega)$ is closed and  $UV^\omega\cap U'{V'}^\omega\neq\emptyset$. 
\end{proof}

\label{bet}
\begin{definition}[B\"{u}chi Effect]
Let $\mathcal{U}$ be an element in $\mathcal{M}_*$ and $\mathcal{V}$ be an element in $\mathcal{M}_{\leq\omega}$. The pair $(\mathcal{U},\mathcal{V})$ is a B\"{u}chi effect of a given expression $e$ if it satisfies:
(a) if $e \downarrow w$, then $w\in \gamma_*(\mathcal{U})$;
(b) if $e \uparrow w$, then $w\in \gamma_{\leq\omega}(\mathcal{V})$.
\end{definition}
Let $\mathcal{V}(\mathfrak{X})$ range over expressions of the form:
$\bigcup_{X\in\mathfrak{X}} (\mathcal{A}_X\cdot X)\ \cup\ \mathcal{B}$
with $\mathcal{A}_X \in \mathcal{M}_*$ and $\mathcal{B}\in \mathcal{M}_{\leq\omega}$. We define the notation $\mathcal{V}(\mathfrak{X}-X)$ and operations on these expressions in the same way as we have done for expressions in the type and effect system in section \ref{sec-e}. The definitions of justifiedness and satisfaction of environments and assignments are adapted to the B\"uchi type system mutatis mutandis. That is, $\Delta$ is justified if 
for all $f\ \&\ (\mathcal{U}, X)$ in $\Delta$ one has
$\Delta\vdash e_f\ \&\ (\mathcal{U}, \mathcal{A}\cdot X\cup \mathcal{V}(\mathfrak{X}-X))$. It is satisfied by $\eta$ if $f\ \&\ (\mathcal{U}, X)$ in $\Delta$ and 
$f\uparrow w$ implies $w\in \gamma_{\leq\omega}(\eta(X))$.

With the above definitions, we introduce the B\"{u}chi type and effect system in Figure \ref{fig-f}.
\begin{figure}[h]
{\scriptsize \[
\inferrule
{ }
{\Delta \vdash_{\mathfrak{A}} o(a)\ \&\ (\alpha_*(\{a\}),\emptyset)}
\quad 
\inferrule
{
{\Delta \vdash_{\mathfrak{A}} e_1\ \&\ (\mathcal{U}_1, \mathcal{V}_1(\mathfrak{X}))}\qquad
{\Delta \vdash_{\mathfrak{A}} e_2\ \&\ (\mathcal{U}_2, \mathcal{V}_2(\mathfrak{X}))}
}
{\Delta \vdash_{\mathfrak{A}} e_1\,;\,e_2\ \&\ (\mathcal{U}_1\cdot \mathcal{U}_2, \mathcal{V}_1(\mathfrak{X}) \cup \mathcal{U}_1\cdot \mathcal{V}_2(\mathfrak{X}))}
\]\[
\inferrule
{
{\Delta \vdash_{\mathfrak{A}} e_1\ \&\ (\mathcal{U}_1, \mathcal{V}_1(\mathfrak{X}))}\qquad
{\Delta \vdash_{\mathfrak{A}} e_2\ \&\ (\mathcal{U}_2, \mathcal{V}_2(\mathfrak{X}))}
}
{\Delta \vdash_{\mathfrak{A}} e_1\,?\,e_2\ \&\ (\mathcal{U}_1\cup \mathcal{U}_2, \mathcal{V}_1(\mathfrak{X})\cup \mathcal{V}_2(\mathfrak{X}))}
\quad
\inferrule
{
}
{\Delta, f\ \&\ (\mathcal{U},X) \vdash_{\mathfrak{A}} f\ \&\ (\mathcal{U},X)}
\]\[
\inferrule
{
{\Delta, f\ \&\ (\mathcal{U}, X) \vdash_{\mathfrak{A}} e_f\ \&\ (\mathcal{U}, \mathcal{A}\cdot X\cup \mathcal{V}(\mathfrak{X}-X))}
}
{\Delta \vdash_{\mathfrak{A}} f\ \&\ (\mathcal{U}, \mathcal{A}^*\cdot \mathcal{V}(\mathfrak{X}-X)\cup \mathcal{A}^\omega)}
\]}
\caption{The B\"{u}chi Type and Effect System}
\label{fig-f}
\end{figure}

By using properties of the B\"{u}chi abstraction in Table \ref{tab-a}, from Theorems \ref{thm-a} and \ref{thm-ca}, we have that this system is  sound and complete.
\begin{theorem}[Soundness]
\label{thm-c}
Given an environment $\Delta$ and an assignment function $\eta$ satisfying that $\eta \models_\mathfrak{A} \Delta$, for all derivations: 
$\Delta\vdash_\mathfrak{A} e\ \&\ (\mathcal{U}, \mathcal{V}(\mathfrak{X}))$
of an expression $e$, we have:
$e\downarrow w$ implies $w \in \gamma_*(\mathcal{U})$
and
$e\uparrow w$ implies $w\in \gamma_{\leq\omega}(\mathcal{V}(\mathfrak{X})_\eta)$.
\end{theorem}
\begin{proof} It follows from the Galois connections in Lemma \ref{lem-galois}. In particular, $U\subseteq \gamma_*(\alpha_*(U))$ and $V\subseteq \gamma_{\leq\omega}(\alpha_{\leq\omega}(V))$. \end{proof}
\begin{theorem}[Completeness]
\label{thm-cb}
Given a non-primitive procedure  $f$, let $T(f)$ be the set of all traces generated by $f$. There is a derivation
$\vdash_\mathfrak{A} f\ \&\ (\mathcal{U}_f, \mathcal{V}_f)$
such that
$T(f) \subseteq L(\mathfrak{A})$ if and only if $\gamma_*(\mathcal{U}_f) \cup \gamma_{\leq\omega}(\mathcal{V}_f)\subseteq L(\mathfrak{A})$.
\end{theorem}
\begin{proof}
Recall the monotone operator $\Phi$ from the proof of Theorem~\ref{thm-ca}. Since $\Phi$ is built up from concatenation and union there is an abstract operator $F$ such that $\alpha_*\,\circ\,\Phi= F\,\circ\,\alpha_*$. Thus $\alpha_*(\textit{lfp}(\Phi))=\textit{lfp}(F)$ and therefore, the judgements $\Delta(\alpha(\vec B))\vdash_\mathfrak{A} e_f\ \&\ (\alpha_*(B_f),\alpha_*(V(\mathfrak X)))$ (again in keeping with the notation of that proof) are derivable in the B\"uchi type system. Using the preservation of $(-)^\omega$ repeatedly, we then obtain the judgements $\vdash_\mathfrak{A} f\ \&\ (\alpha_*(U_f),\alpha_{\leq\omega}(V_f))$ where $U_f=\{w\mid f\downarrow w\}$ and $V_f=\{w\mid f\uparrow w\}$. Letting $\mathcal{U}_f=\alpha_*(U_f)$ and $\mathcal{V}_f=\alpha_{\leq\omega}(V_f)$ the claim then follows using Lemma~\ref{compless}. 
\end{proof}
\subsection{Type inference and complexity}\label{compl}
Given that the abstract lattices and thus the set of types is finite,
type inference is a standard application of well-known techniques. We
therefore just sketch it here to give an idea of the complexity.

From a given program we can construct in linear time a skeleton typing
derivation for the finitary effect annotations. The skeleton typing
derivation contains variables in place of actual effect annotations;
the number of these variables is linear in the program size. The
side conditions of the typing rules then become constraints on these
variables and any solution will yield a valid typing derivation. In
quadratic time (assuming that $\mathcal{M}_{\leq\omega}$ has constant size) we
can then compute the least solution of these constraints using the
usual iteration algorithms known from abstract interpretation. Once we
have in this way obtained the finitary effect annotations we can then
(in linear time) derive the infinitary ones using the $(-)^\omega$ and
infinitary concatenation operators on $\mathcal{M}_{\leq\omega}$.

Once the type of an expression has been found one can then check (in constant time) whether the language denoted by it is accepted by the policy automaton. 

If we are interested in complexity as a function of the size of the
policy automaton the situation is of course different. The important
parameter here is the size of the abstract lattices since the number
of iterations as well as the runtime of the algorithms for computing
the abstractions of concatenation, union, infinite iteration are
linear in this parameter. If $n$ is the number of states of the policy
automaton then the number of classes can be bounded by $2^{2n^2}$
since each class is characterised by two sets of pairs of states. The
resulting exponential in $n$ runtime of our algorithms is no surprise
since the PSPACE-complete problem of universality of B\"uchi automata
is easily reduced to type checking. We believe that by clever space
management our algorithms can be implemented in polynomial space but
we have not verified this.

On a positive note we remark that for a small policy automaton the set
of classes is manageable as we see in the examples below. We also note
that once the classes have been computed and the abstract functions
tabulated one can then analyse many programs of arbitrary size.

\section{Conclusions}
We have developed a type-and-effect system for capturing possibly
infinite traces of recursively defined first-order procedures. The
type-and-effect system is sound and complete with respect to inclusion of traces
in a given B\"uchi (``policy'') automaton. The effect annotations are
from a finite set that can be effectively computed from the B\"uchi
automaton. Type inference using constraint solving is thus possible. 
We emphasize that the resulting ability to decide satisfaction of temporal properties of traces is not claimed as a new result here; since it has long been known in the context of model checking. The novelty lies in the presentation as a type and effect system that follows the standard pattern of such systems. As we explain below, this opens the way for smooth integration with existing type-theoretic technology. 

The proofs of soundness and completeness are organised in a modular
fashion and decomposed into a type-theoretic part expressed in the
form of an infinitary system (Section~\ref{sec-e}) and a
lattice-theoretic part (Section~\ref{ba}). Concretely, this Section
defines an abstract domain from any given B\"uchi automaton and
derives crucial properties of this abstraction. We consider a
contribution of independent interest. The finite part of this abstraction, i.e., $\mathcal{M}_*$, is akin to the abstract domain proposed by Cousot et al \cite{cousot:dblp:c13} which also has finite abstraction values and its abstraction function preserves the least fixed point as well. The infinite part $\mathcal{M}_{\leq\omega}$ is a new abstract domain with its abstraction function preserving not only  least fixed points but also the new operator $(-)^\omega$. We remark here that the abstraction function does not preserve 
greatest fixed points so that the introduction of the $(-)^\omega$ operator on the abstract domain is a necessary device. This extension makes the B\"uchi type and effect system powerful enough to capture and reason about infinitary properties like liveness and fairness.

We have sketched a combination of our simple type system with a generic type and effect system for class-based object-oriented  languages. Other possible 
 extensions are in the direction of effectful functional
programming. The standard notation for type-and-effect systems as
described e.g.\ in Henglein and Niss' survey \cite{attapl} could be
used for our effect system mutatis mutandis leading in particular to
function types like $w\stackrel{\epsilon}{\rightarrow}w'$ where
$w,w'$ are types and $\epsilon$ is a B\"uchi effect (element of
our abstract lattice) describing the latent effect of a
function. Assuming that we only allow first-order recursive
definitions the design of such a type system would be completely
standard. For higher-order recursion some extra technical work would
be needed to lift the last rule from Fig.~\ref{fig-e} and its
corresponding abstraction to this case. 

It is this option of integration with expressive type systems that
makes our abstraction so attractive and superior (in this context!) to
classical  methods based on model checking. 








\appendix
\section{Examples}
\begin{example}
Consider the following definition:
\begin{eqnarray*}
f\ =\ o(b)\ ;\ o(a)\ ;\ f
\end{eqnarray*}
Suppose that we want to verify the property:
\begin{quote}
{\it Every finite trace generated by $f$ ends with $b$. Every infinite trace generated by $f$ contains infinite many $b$s.}
\end{quote}
We can use the following extended B\"{u}chi automaton $\mathfrak{A}$ to formalize this property. 
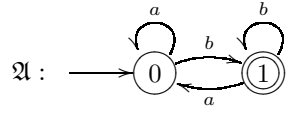
\begin{eqnarray*}
\mathfrak{A} : \xymatrix{ \ar[r] & *++[o][F]{0} \ar@(ur,ul)_{a}\ar@/^/[r]^b  & 
*++[o][F=]{1} \ar@(ur,ul)_{b}\ar@/^/[l]^a}
\end{eqnarray*}
We have:
\begin{eqnarray*}
L(\mathfrak{A}) = (a^*b^+)^+ \cup (a^*b)^\omega\ .
\end{eqnarray*}
By the definition of $\sim$, we have that $\mathcal{Q}$ consists of
four equivalence classes: the set of empty word ($[\epsilon] =
\{\epsilon\}$), the set of non-empty words consists of $a$ ($[a] =
a^+$), the set of words ending with $a$ and containing at least one
$b$ ($[ba] = (a+b)^*a - a^+ = (a+b)^*b(a+b)^*a$), the set of words
ending with $b$ ($[b] = (a+b)^*b$).
Further, the set $\mathcal{C} = \{(C,D)\mid C,D\in \mathcal{Q}\wedge CD=C, DD = D\}$ is as follows:
\begin{eqnarray*}
& & \{
([\epsilon], [\epsilon]),
([a], [\epsilon]),
([a], [a]),
([b], [\epsilon]),
\\
& & \ \ 
([b], [b]),
([ba], [\epsilon]),
([ba], [a]),
([ba], [ba])
\}
\end{eqnarray*}
and abstractions of $L(\mathfrak{A})$ are:
\begin{eqnarray*}
& & \mathcal{U}_L = \{[b]\}\in \mathcal{M}_*\\
& & \mathcal{V}_L =
\{
([b],[\epsilon]),
([b],[b]),
([ba],[ba])
\}\in \mathcal{M}_{\leq\omega}\ .
\end{eqnarray*}
By using the B\"uchi type and effect system defined in Figure \ref{fig-f}, we have that $(\mathcal{U}, \mathcal{V}) = (\emptyset, (\{[ba]\})^\omega)$ is a B\"uchi effect of $f$. Further, 
\begin{eqnarray*}
\mathcal{V} = \alpha_{\leq\omega}((\gamma_{*}(\{[ba]\}))^\omega)
&=& \{([b], [b]), ([ba], [ba])\}\ .
\end{eqnarray*}
So, $T(e_f) \subseteq \gamma_{\leq\omega}(\mathcal{V})\subseteq \gamma_{\leq\omega}(\mathcal{V}_L)$.
That is, all traces generated by $f$ satisfy the target property. This coincides with our observation that $f$ does not generate any finite traces and the only infinite trace generated by $f$ is $(ba)^\omega$ which contains infinite many $b$s. 
\end{example}
\begin{example}
Consider the following C-like program: 
\begin{verbatim}
0   #define TIMEOUT 65536 
1   while (true) { 
2     i = 0;
3     while (i++ < TIMEOUT && s != 0) {
4       unsigned int s = auth(); /* o(a); */
5     }  /* o(c); */
6     work(); /* o(b); */
7   }
\end{verbatim}
We would like to verify that line 6 is executed infinitely often under the fairness assumption that the while loop 3 always terminates. To this end, we can annotate the above program by uncommenting the event-issuing commands
and abstract the so annotated program as the definition: 
\begin{eqnarray*}
& & f\ =\ g\,;\,o(b)\,;\,f\\
& & g\ =\ (o(a)\,;\,g)\ ?\ o(c) 
\end{eqnarray*}
We are then interested in the property ``infinitely many $b$'' assuming that ``infinitely often $c$'' (fairness) or equivalently: ``infinitely many $b$ or finitely many $c$.'' This property can be readily expressed as the following B\"uchi automaton 
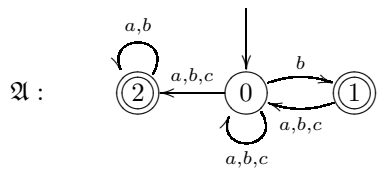
$$
\xymatrix{&&\ar[d]\\\mathfrak{A} : &
*++[o][F=]{2}\ar@(ur,ul)_{a,b}&*++[o][F]{0}\ar[l]_{a,b,c}\ar@(dr,dl)^{a,b,c}\ar@/^/[r]^b & *++[o][F=]{1} \ar@/^/[l]^{a,b,c}&}
$$
By the definition of $\sim$, we have that the set $\mathcal{Q}$ consists of the following equivalence classes:
\begin{eqnarray*}
& & [\epsilon] = \{\epsilon\}\qquad [a] = \{a\}\qquad [b] = \{b\}\qquad [c] = \{c\} \\
& & [aa] = a^+a \qquad\qquad \qquad \quad [ba] = (a+b)^+a - [aa]\\
& & [ab] = a^+b \qquad\qquad \qquad \quad\,[bb] = (a+b)^+b - [ab]\\
& & [cb] = (a+c)^*c(a+c)^*b\\
& & [bcb] = (a+b+c)^*c(a+b+c)^*b - [cb]\\
& & [cca] = (a+c)^+c\ \cup\ (a+c)^* c (a+c)^*a\\
& & [bca] = (a+b+c)^+c\ \cup\ \\ 
& & \qquad\quad\ \ (a+b+c)^* c (a+b+c)^*a - [cca]\ .
\end{eqnarray*}
Further, the set $\mathcal{C}$ consists of the following pairs:
\begin{eqnarray*}
& & ([\epsilon],[\epsilon])\ \ 
\quad\quad([a],[\epsilon])\ \ 
o\quad\quad([b],[\epsilon])\ \ 
\quad\ \ \ \, ([c],[\epsilon])
\\
& & ([aa],[\epsilon])\ \ 
\quad \ \, ([ba],[\epsilon])\ \ 
\quad \ \, ([ab],[\epsilon])\ \ 
\quad \ \ ([bb],[\epsilon])
\\
& & ([cb],[\epsilon])\quad\ \ \,\ \ 
([bcb],[\epsilon])\quad\ \ 
([cca],[\epsilon])\quad\, \ \ 
([bca],[\epsilon])
\\
& & ([aa],[aa])\ \ \ \ \ 
([ba],[aa])\ \ \,\ \ 
([ba],[ba])\ \ \ \ \,\ \ 
([bb],[bb]) 
\\
& & ([bcb],[bb])\ \ \ \ \ 
([bcb],[bcb])\,\ \ 
([cca],[aa])\ \, \ \ 
([cca],[cca])\, 
\\
& &([bca],[aa])\ \ \ \ 
([bca],[ba])\ \ \ 
([bca],[cca])\  \ \ 
([bca],[bca])
\end{eqnarray*}
Then, the abstractions of $L(\mathfrak{A})$ are as follows:
\begin{eqnarray*}
& & \mathcal{U}_L = \mathcal{Q} - \{[\epsilon]\}\in \mathcal{M}_*\\
& & \mathcal{V}_L =
\mathcal{C}-
\{
([\epsilon],[\epsilon]), ([cca],[cca]), ([bca],[cca])
\}\in \mathcal{M}_{\leq\omega}\ .
\end{eqnarray*}
By using the B\"uchi type and effect system, we get the effect of $g$ as the pair:
$$(\mathcal{U}_g, \mathcal{V}_g) 
= (
\{[c], [cca]\},
\{
([aa], [aa])\}
)\ .$$ 
Then, the effect of $f$ is the pair $(\mathcal{U}_f, \mathcal{V}_f)$ given as follows: 
$$(\emptyset,\{
([aa],[aa]),
([bca],[aa]),
([bcb],[bcb]),
([bca],[bca])
\})\ .$$ 
Since $\mathcal{U}_f\subseteq\mathcal{U}_L$ and $\mathcal{V}_f\subseteq\mathcal{V}_L$, we have that the program satisfies the property. 
\end{example}

\section{Proofs}

\label{a-p}
\begin{proof}
(of Theorem \ref{thm-a})
We  proceed by induction on the structure of the type and effect system. The only interesting case is that for the last rule in Figure \ref{fig-e}. Let us show that if
\begin{eqnarray*}
f \uparrow w
\mbox{\quad and \quad} 
\Delta \vdash f\ \&\ (U, A^*\cdot V(\mathfrak{X}-X)\cup A^\omega)\ ,
\end{eqnarray*}
then $w$ is in $A^*\cdot V(\eta)\cup A^\omega$. We introduce the assignment functions:
\begin{eqnarray*}
\eta_n = \left\{ 
\begin{array}{ll}
\eta[X\mapsto \Sigma^{\leq\omega}] & \mbox{\quad if\ \ $n = 0$}\ ;\\
\eta[X\mapsto A \cdot X_{\eta_{n-1}} \cup\ V(\eta)]& \mbox{\quad if\ \ $n\geq 1$}\ . 
\end{array}
\right.
\end{eqnarray*}
From Lemma \ref{lem-b}, $\eta_0\models \Delta, f\ \&\ (U,X)$. Assume that $\eta_n\models \Delta, f\ \&\ (U,X)$. By inductive hypothesis, we have that for all $e_f\uparrow w$, $w$ is in 
\begin{eqnarray*}
A\cdot X_{\eta_n}\ \cup\ V(\mathfrak{X}-X)_{\eta_n} = A\cdot X_{\eta_n}\ \cup\ V(\mathfrak{X}-X)_{\eta} = X_{\eta_{n+1}}\ .
\end{eqnarray*}
Further, by Lemma \ref{lem-b}, $\eta_{n+1}\models \Delta, f\ \&\ (U,X)$. By mathematical induction, we have that $\eta_n\models \Delta, f\ \&\ (U,X)$ for all natural numbers $n$. Define $\eta_\omega$ as 
\begin{eqnarray*}
\eta[X\mapsto \bigcap_{n = 0}^{\infty}X_{\eta_n}]\ .
\end{eqnarray*}
We get $\eta_\omega\models \Delta, f\ \&\ (U,X)$. Notice that $\eta_\omega(X)$ is equal to $A^*\cdot V(\eta)\ \cup\ A^\omega$. By the definition of $\models$, $w$ is in $A^*\cdot V(\eta)\cup A^{\omega}$.
\end{proof}
\begin{proof}
(of Lemma \ref{lem-d})
a) and b) are obvious from the definition. Property c) coincides with b) when $D=\{\epsilon\}$. Otherwise, let $w,w'\in CD^\omega$. Decompose $w=w_1w_2w_3\dots$ and $w'=w_1'w_2'w_3'\dots$ so that $w_1,w_1'\in C$ and $w_i,w_i'\in D$ for $i>1$. 
We have $w_i \sim w_i'$ for all $i$ so any accepting run for $w$ yields an accepting run for $w'$ by the definition of $\sim$. 
Property d) is again trivial when $w\in\Sigma^*$ (choose $C=[w]$ and $D=\{\epsilon\}$) and otherwise appears already in B\"uchi's work. For the record, write $w=a_1a_2a_3\dots $ with $a_i\in\Sigma$ and ``colour'' the set $\{i,j\}$ with $i<j$ with the $\sim$-class of $a_ia_{i+1}\dots a_{j-1}$. By Ramsey's theorem there exists an infinite set of indices $i_1<i_2<i_3\dots$ and a class $D$ so that $[a_{i_k}\dots a_{i_{k+1}-1}]=D$ for all $k>0$. The claim follows with $C:=[a_1\dots a_{i_2-1}]$. 
\end{proof}
\begin{proof}
(Proof of Lemma \ref{lem-in})
  This is trivial for $\mathcal{M}_*$ which is a powerset lattice. As for
  $\mathcal{M}_{\leq\omega}$, one must show that unions and
  intersections of closed sets are again closed. So, let
  $(\mathcal{V}_i)_{i\in I}$ be a family of closed sets. To argue that
  the union of this family is closed, suppose that $UV^\omega\cap
  CD^\omega\neq\emptyset$ for some $(C,D)$ contained in that
  union. Then, $(C,D)\in\mathcal{V}_i$ for some $i$, so
  $(U,V)\in\mathcal{V}_i$. Since $\mathcal{V}_i$ is closed,
  $(U,V)$ is contained in the union. As for the intersection, suppose
  that $UV^\omega\cap CD^\omega\neq\emptyset$ for some $(C,D)$
  contained in the intersection. Then, $(C,D)\in\mathcal{V}_i$ for all
  $i$, so $(U,V)\in\mathcal{V}_i$ for all $i$. Since each
  $\mathcal{V}_i$ is closed, $(U,V)$ is contained in the
  intersection, too.
\end{proof}
\begin{proof}
(Proof of Lemma \ref{lem-con})
We need to show that  $\mathcal{U}\cdot\mathcal{V}$ is closed so suppose that 
$UV^\omega\cap ACD^\omega\neq\emptyset$ where $A\in \mathcal{U}$ and $(C,D)\in\mathcal{V}$. We may assume that $V\neq[\epsilon]$ for otherwise the claim is trivial. Decomposing the witnessing word, we get finite words $x,y$ where
 $xy\in U$ and $x\in A$. Note that $UV=U$.
 Thus, $U=AY$ with $Y$ the class of $y$. As a result, $YV^\omega\cap CD^\omega\neq\emptyset$, so $(YV,V)\in\mathcal{V}$ by closedness and finally, $(U,V)\in \mathcal{U}\cdot\mathcal{V}$ since $U=AYV$. 
\end{proof}

\section{Region-Based B\"uchi Type and Effect System}
\label{a-os}
In order to make the B\"uchi type and effect system
given in Section \ref{bet} more functional and effective in programming practices,
by integrating with 
region types \cite{Luc88,attapl,comlan}, 
we extend it to a region-based B\"uchi type and effect system
for Featherweight Java with field update.
Based on this, the future goal is to
develop and implement a powerful type system for Java-like languages
in which ($\omega$)-regular properties of traces
can be properly characterized for verification purposes.

\subsection{Syntax}
The syntax of an expression $e$ in Featherweight Java with field update 
is given as follows:
\[x \in Var\quad f\in Fld\quad m \in Mtd \quad C,D \in Cls\]
\vskip-0.5cm
\begin{eqnarray*}
e &::=& {\bf o}(a)\ |\ null\ |\ x\ |\ {\bf new}\ C\ |\ x.f\ |\ x.f := y\ |\\
& & x.m(\overline{y})\ |\ {\bf let}\ x = e_1\ {\bf in}\ e_2\ |\ {\bf if}\ x = y\ {\bf then}\ e_1\ {\bf else}\ e_2
\end{eqnarray*}
For the sake of simplicity, we omit primitive types and casting and 
assume that every expression is in let normal form.
In the definition of the {\bf if}-{\bf then}-{\bf else} expression,
the expression $x = y$ denotes an unusual judgement between objects
which is independent on booleans.
The notation $\overline{y}$ denotes a sequence of variables.

Additionally, the expression $\mathbf{o}$ is used to 
produce appealing events.
It is a global primitive procedure which
is not part of Featherweight Java with filed update.
It is added as annotations to programs
for the purposes of property characterization.

Let $\preceq\,\in\,\mathcal{P}(Cls\times Cls)$
be the subclass relation between classes. 
Let $fields \in Cls\to \mathcal{P}(Fld)$ and
$methods\in Cls \to \mathcal{P}(Mtd)$
be mappings from a class to its fields and methods respectively.
We use $mtable\in Cls\times Mtd \rightharpoonup \overline{Var}\times Expr$ to denote the method table 
which assigns to each method its definition, i.e.,
its formal parameters (a sequence of variables) and
its body (an expression).
With these definitions, a program $P$ is given as follows: 
\[P = (\preceq, fields, methods, mtable)\]
Usual well-formedness conditions on methods are assumed to ensure
the inheritance relation between same methods from different classes.

\subsection{Operational Semantics}
The operational semantics of an expression $e$ are given in form of:
\begin{center}
\framebox[\width]{\quad $(s,h)\vdash e \Downarrow v, h'\ \&\ w$\quad
\par
\quad $(s,h)\vdash e \Uparrow\ \&\ w$\quad
}
\end{center}
We use $e\Downarrow v$ to denote that the evaluation of $e$ terminates and the value is $v$.
The notation $e\Uparrow\ $ means that the evaluation of $e$ doesn't terminate.
A value $v\in Val$ of an expression is a location $l\in Loc$ or $null$.
A state $(s,h)$ is consisted of 
a stack $s\in Var \rightharpoonup Val$ which is a partial function assigning to each variable a value 
and a heap $h\in Loc \rightharpoonup Cls\times (Fld\rightharpoonup Val)$ which is a partial function assigning to each location a pair of a class and
values assigned to fields of this class. 
In addition, as side effects, 
an infinite or a finite trace $w$ which is generated from the set $\Sigma$ of events
by ordinary concatenations
is attached to each evaluation. 
The whole operational semantics is given as follows:
\[
\inferrule*[Left={Prim-a}]
{ }
{(s,h)\vdash {\bf o}(a) \Downarrow null, h\ \&\ a}
\]\[
\inferrule*[Left={Null}]
{ }
{(s,h)\vdash null \Downarrow null, h\ \&\ \epsilon}
\qquad\qquad\qquad 
\inferrule*[Left={Var}]
{ }
{(s,h)\vdash x \Downarrow s(x), h\ \&\ \epsilon}
\]\[
\inferrule*[Left={New}]
{l \not\in dom(h)\qquad F = [f\mapsto null]_{f\in fields(C)}}
{(s,h)\vdash {\bf new}\ C \Downarrow l, h[l\mapsto (C,F)]\ \&\ \epsilon}
\]\[
\inferrule*[Left={Get}]
{s(x) = l\qquad h(l) = (C,F)}
{(s,h)\vdash x.f \Downarrow F(f), h\ \&\ \epsilon}
\]\[
\inferrule*[Left={Set}]
{s(x) = l\qquad h(l) = (C,F)\\\\ 
h' = h[l\mapsto (C, F[f\mapsto s(y)])]}
{(s,h)\vdash x.f := y \Downarrow s(y), h'\ \&\ \epsilon}
\]\[\inferrule*[Left={Call-a}]
{s(x) = l\qquad h(l) = (C,F)\\\\
mtable(C,m) = (\overline{x}, e) \qquad |\overline{x}| = |\overline{y}| = n\\\\
s' = [this\mapsto l]\cup[x_i\mapsto s(y_i)]_{i\in\{1,2,\dots,n\}}\\\\
(s',h)\vdash e \Downarrow v, h'\ \&\ w}
{(s,h)\vdash x.m(\overline{y}) \Downarrow v, h'\ \&\ w}
\]\[
\inferrule*[Left={Let-a}]
{(s,h)\vdash e_1 \Downarrow v_1, h_1\ \&\ w_1\\\\
(s[x\mapsto v_1],h_1)\vdash e_2 \Downarrow v_2, h_2\ \&\ w_2}
{(s,h)\vdash {\bf let}\ x = e_1\ {\bf in}\ e_2 \Downarrow v_2, h_2\ \&\ w_1\cdot w_2}
\]\[
\inferrule*[Left={If-a}]
{s(x) = s(y)\qquad (s,h)\vdash e_1 \Downarrow v, h'\ \&\ w}
{(s,h)\vdash {\bf if}\ x = y\ {\bf then}\ e_1\ {\bf else}\ e_2 \Downarrow v, h'\ \&\ w}
\]\[\inferrule*[Left={If-b}]
{s(x) \neq s(y)\qquad (s,h)\vdash e_2 \Downarrow v, h'\ \&\ w}
{(s,h)\vdash {\bf if}\ x = y\ {\bf then}\ e_1\ {\bf else}\ e_2 \Downarrow v, h'\ \&\ w}
\]\[
\inferrule*[Left={Epsilon}]
{ }
{(s,h)\vdash e \Uparrow \ \&\ \epsilon}
\qquad\qquad\qquad
\inferrule*[Left={Prim-b}]
{ }
{(s,h)\vdash {\bf o}(a) \Uparrow \ \&\ a}
\]\[
\inferrule*[Left={Call-b}]
{s(x) = l\qquad h(l) = (C,F)\\\\
mtable(C,m) = (\overline{x}, e) \qquad |\overline{x}| = |\overline{y}| = n\\\\
s' = [this\mapsto l]\cup[x_i\mapsto s(y_i)]_{i\in\{1,2,\dots,n\}}\\\\
(s',h)\vdash e \Uparrow \ \&\ w}
{(s,h)\vdash x.m(\overline{y}) \Uparrow\ \&\ w}
\]\[
\inferrule*[Left={Let-b}]
{(s,h)\vdash e_1 \Downarrow v_1, h_1\ \&\ w_1\\\\
(s[x\mapsto v_1],h_1)\vdash e_2 \Uparrow \ \&\ w_2}
{(s,h)\vdash {\bf let}\ x = e_1\ {\bf in}\ e_2 \Uparrow \ \&\ w_1\cdot w_2}
\]\[
\inferrule*[Left={Let-c}]
{(s,h)\vdash e_1 \Uparrow \ \&\ w}
{(s,h)\vdash {\bf let}\ x = e_1\ {\bf in}\ e_2 \Uparrow \ \&\ w}
\]\[
\inferrule*[Left={If-c}]
{s(x) = s(y)\qquad (s,h)\vdash e_1 \Uparrow \ \&\ w}
{(s,h)\vdash {\bf if}\ x = y\ {\bf then}\ e_1\ {\bf else}\ e_2 \Uparrow \ \&\ w}
\]\[
\inferrule*[Left={If-d}]
{s(x) \neq s(y)\qquad (s,h)\vdash e_2 \Uparrow\ \&\ w}
{(s,h)\vdash {\bf if}\ x = y\ {\bf then}\ e_1\ {\bf else}\ e_2 \Uparrow \ \&\ w}
\]
Equipped with the rule {\sc Epsilon},
all finite prefixes of an infinite trace produced 
by a non-terminate evaluation are captured. 
As for other rules, they are defined in the usual way.

\subsection{Region-Based B\"uchi Type and Effect System}
We now sketch an integration of the  B\"uchi type and effect system with
the region type system for Java given by Beringer et al \cite{comlan}. 
Let us first explain why this integration is interesting and useful
by an example.
Considering the following fragment of Java-like code:
\begin{verbatim}
   class C { 
     void f (String arg);
   }
\end{verbatim}
It could be refined using two different regions \texttt{r} and \texttt{r'}  with B\"uchi effects 
\texttt{(U,V)} and \texttt{(U',V')} respectively as follows:
\begin{verbatim}
   class C@r { 
     void f (String@X arg) & (U,V); 
   }
   class C@r' { 
     void f (String@X' arg) & (U',V'); 
   }
\end{verbatim}
Then, an object \texttt{o} typed \texttt{C@r} expects a \texttt{String@X} as
argument to \texttt{f} and \texttt{o.f()} will exhibit a \texttt{(U,V)} effect. An
object \texttt{o1} typed \texttt{C@r'} expects a \texttt{String@X'} as
argument to \texttt{f} and \texttt{o1.f()} will exhibit a \texttt{(U',V')} effect. In this particular case, regions denote locations at which effects are produced.

Generally, a region $r\in Reg$ is a static abstraction of concrete locations
which can be considered as a set of concrete locations.
A class type $C$ can then be equipped with a set $R$ of regions,
yielding a refined type $C_R$ that places the constraint that its members belong to one of the regions in $R$. We summarize these definitions and introduce new varaiabls as follows: 
$$R, S\in \mathcal{P}(Reg)\quad
C_R, \tau, \sigma\in(Cls\times\mathcal{P}(Reg)) \uplus \{unit\} = Typ$$
Here, the $unit$ type is introduced for 
typing the expression $\mathbf{o}(a)$.
It is easy to define the subtype relation between region-based types:
$C_R <: C'_{R'}$ if and only if $C \preceq C' \wedge R \subseteq R'$.
and to extend this definition to sequences of types as follows:
\[\overline{\sigma} <: \overline{\sigma}'\quad\Leftrightarrow\quad |\overline{\sigma}| = |\overline{\sigma}'|\ \ \wedge\ \ \forall i\in |\overline{\sigma}|\,.\,\sigma_i <: \sigma'_i\]

With respect to the subtype relation, the following field typings: 
\[A^{set}, A^{get} \in Cls \times Reg \times Fld \to Typ\]
assign to each field in each region-annotated class respectively a set-type
which is a contravariant type for data written to the field and
a get-type which is a covariant type for data read from the field.
The following well-formedness conditions on $A^{set}$ and $A^{get}$ are imposed:
\[A^{set}(C,r,f) <: A^{get}(C,r,f)\]
and if $D\preceq C$, then 
\[A^{set}(C,r,f) <: A^{set}(D,r,f)
\ \  \wedge\ \ A^{get}(D,r,f) <: A^{get}(C,r,f)\]

Given a B\"uchi automaton $\mathfrak{A}$, 
let $\mathcal{M}_*$ and $\mathcal{M}_{\leq\omega}$ 
be the B\"uchi abstractions 
for $\Sigma^*$ and $\Sigma^{\leq\omega}$ respectively, 
as  defined in Section \ref{ba}. We define effects as: 
$Ef\!f = \mathcal{M}_*\times\mathcal{M}_{\leq\omega}$.
Then, the following typing rule:
\[M \in Cls \times Reg \times Fld \to \overline{Typ}\times Typ \times Ef\!f\]
assigns to each method in each region-annotated class a functional type
with an effect. The subtype relation 
\[\overline{\sigma}\stackrel{(\mathcal{U}, \mathcal{V})}{\longrightarrow} \tau\quad <:\quad \overline{\sigma}'\stackrel{(\mathcal{U}', \mathcal{V}')}{\longrightarrow} \tau'\] 
between effect annotated functional types is defined as: 
\[\overline{\sigma}' <: \overline{\sigma}\ \wedge\ \tau <: \tau'\ \wedge\ \mathcal{U}\subseteq \mathcal{U}'\ \wedge\ \mathcal{V}\subseteq \mathcal{V}'\]
That is, the input type is contravariant and the output type is covariant.
The well-formedness condition on $M$ is:
\[D \preceq C\quad\Rightarrow \quad M(D,r,m) <: M(C,r,m)\]
for all classes, regions, and methods.

With the above definitions, combinined with the type and effect system given in Section \ref{ba}, we arrive at the region-based B\"uchi type and effect system as follows: 
\[
\inferrule*[Left={T-Sub}]
{\Gamma \vdash_\mathfrak{A} e: \tau \ \&\ (\mathcal{U}, \mathcal{V}(\mathfrak{X})) \\\\ 
\tau <: \tau'\qquad\qquad
\mathcal{U}\subseteq \mathcal{U'}\qquad\qquad
\mathcal{V}(\mathfrak{X}) \sqsubseteq \mathcal{V'}(\mathfrak{X})}
{\Gamma \vdash_\mathfrak{A} e:\tau' \ \&\ (\mathcal{U'}, \mathcal{V'}(\mathfrak{X}))}
\]\[
\inferrule*[Left={T-Prim}]
{ }
{\Gamma \vdash_\mathfrak{A} {\bf o}(a) : unit \ \&\ (\alpha_*(\{a\}), \emptyset)}
\]\[
\inferrule*[Left={T-Null}]
{ }
{\Gamma \vdash_\mathfrak{A} null : C_\emptyset \ \&\ (\alpha_*(\{\epsilon\}), \emptyset)}
\]\[\inferrule*[Left={T-Var}]
{ }
{\Gamma, x:\tau \vdash_\mathfrak{A} x:\tau \ \&\ (\alpha_*(\{\epsilon\}), \emptyset)}
\]\[
\inferrule*[Left={T-New}]
{ }
{\Gamma \vdash_\mathfrak{A} {\bf new}\ C : C_{\{r\}} \ \&\ (\alpha_*(\{\epsilon\}), \emptyset)}
\]\[
\inferrule*[Left={T-Get}]
{ \forall r\in R\,.\,A^{get}(C,r,f) <: \tau}
{\Gamma, x: C_R \vdash_\mathfrak{A} x.f : \tau \ \&\ (\alpha_*(\{\epsilon\}), \emptyset)}
\]\[
\inferrule*[Left={T-Set}]
{ \forall r\in R\,.\,\tau <: A^{set}(C,r,f)}
{\Gamma, x: C_R, y: \tau \vdash_\mathfrak{A} x.f := y : \tau \ \&\ (\alpha_*(\{\epsilon\}), \emptyset)}
\]\[
\inferrule*[Left={T-Call}]
{ \overline{\sigma}_r\stackrel{(\mathcal{U}_r, \mathcal{V}_r)}{\longrightarrow} \tau_r = M(C,r,m) \\\\
\forall r\in R\,.\, \overline{\sigma}_r\stackrel{(\mathcal{U}_r, \mathcal{V}_r)}{\longrightarrow} \tau_r\ \ <:\ \ \overline{\sigma}\stackrel{(\mathcal{U}, \mathcal{V})}{\longrightarrow} \tau}
{\Gamma, x : C_R, \overline{y}:\overline{\sigma} \vdash_\mathfrak{A} x.m(\overline{y}) : \tau \ \&\ (\mathcal{U}, \cup_{r\in R}\{X_r\})}
\]\[
\inferrule*[Left={T-Let}]
{\Gamma \vdash_\mathfrak{A} e_1 : \tau_1\ \&\ (\mathcal{U}_1, \mathcal{V}_1(\mathfrak{X}))\\\\
\Gamma, x:\tau_1 \vdash_\mathfrak{A} e_2 : \tau_2\ \&\ (\mathcal{U}_2, \mathcal{V}_2(\mathfrak{X}))}
{\Gamma \vdash_\mathfrak{A} {\bf let}\ x = e_1\ {\bf in}\ e_2 : \tau_2 \ \&\ (\mathcal{U}_1\cdot \mathcal{U}_2, \mathcal{V}_1(\mathfrak{X})\cup \mathcal{U}_1\cdot \mathcal{V}_2(\mathfrak{X}))}
\]\[
\inferrule*[Left={T-If}]
{\Gamma, x:C_{R\cap S}, y: D_{R\cap S} \vdash_\mathfrak{A} e_1 : \tau\ \&\ (\mathcal{U}_1, \mathcal{V}_1(\mathfrak{X}))\\\\
\Gamma, x:C_R, y:D_S \vdash_\mathfrak{A} e_2 : \tau\ \&\ (\mathcal{U}_2, \mathcal{V}_2(\mathfrak{X}))}
{\Gamma, x:C_R, y:D_S \vdash_\mathfrak{A} \\
{\bf if}\ x = y\ {\bf then}\ e_1\ {\bf else}\  e_2 : \tau \ \&\ (\mathcal{U}_1\cup \mathcal{U}_2, \mathcal{V}_1(\mathfrak{X})\cup \mathcal{V}_2(\mathfrak{X}))}
\]
The typing judgement for expressions $e$ is
\begin{center}
\framebox[\width]{\quad 
$\Gamma \vdash_\mathfrak{A} e:\tau \ \&\ (\mathcal{U}, \mathcal{V}(\mathfrak{X}))$
}
\end{center}
with $\Gamma$ the type environment, $\mathfrak{A}$ the policy B\"uchi automaton, $\tau$ the type of $e$, $\mathcal{U}$ the set of finite traces, and $\mathcal{V}(\mathfrak{X})$ the expression for infinite traces.
Among the typing rules, is the following rule:
\begin{mathpar}
\inferrule*[Left={T-Sub}]
{\Gamma \vdash_\mathfrak{A} e: \tau \ \&\ (\mathcal{U}, \mathcal{V}(\mathfrak{X})) \\\\ 
\tau <: \tau'\qquad\qquad
\mathcal{U}\subseteq \mathcal{U'}\qquad\qquad
\mathcal{V}(\mathfrak{X}) \sqsubseteq \mathcal{V'}(\mathfrak{X})}
{\Gamma \vdash_\mathfrak{A} e:\tau' \ \&\ (\mathcal{U'}, \mathcal{V'}(\mathfrak{X}))}
\end{mathpar}
where
\[\mathcal{V}(\mathfrak{X}) \sqsubseteq \mathcal{V}'(\mathfrak{X}) \quad\Leftrightarrow\quad \forall \eta\in \mathfrak{X} \to \mathcal{M}_{\leq\omega}\,.\, \mathcal{V}(\eta) \subseteq \mathcal{V}'(\eta)\]
Notice that in Section \ref{ba}, there are two typing rules for  function calls. That is, one rule is used to directly get the effect if the effect has been assumed in the environment and another rule is used to derive the effect with an effect assumption added into the environment. However, in order to integrate with region-based type systems, in our region-based B\"uchi type and effect system, we only use the following rule: 
\begin{mathpar}
\inferrule*[Left={T-Call}]
{ \overline{\sigma}_r\stackrel{(\mathcal{U}_r, \mathcal{V}_r)}{\longrightarrow} \tau_r = M(C,r,m) \\\\
\forall r\in R\,.\, \overline{\sigma}_r\stackrel{(\mathcal{U}_r, \mathcal{V}_r)}{\longrightarrow} \tau_r\ \ <:\ \ \overline{\sigma}\stackrel{(\mathcal{U}, \mathcal{V})}{\longrightarrow} \tau}
{\Gamma, x : C_R, \overline{y}:\overline{\sigma} \vdash_\mathfrak{A} x.m(\overline{y}) : \tau \ \&\ (\mathcal{U}, \cup_{r\in R}\{X_r\})}
\end{mathpar}
That is, the set $\mathcal{U}$ of finite traces
is taken from the declarations of finite traces in $M$.
As for the set of infinite traces,
we only put a set of placeholders $X_r$.

Further, a program $P$ is well-typed if and only if for all classes $C$, regions $r$, and methods $m$ such that  
\[mtable(C,r,m) = (\overline{x}, e)\ \wedge\ M(C,r,m) = \overline{\sigma}_r\stackrel{(\mathcal{U}_r, \mathcal{V}_r)}{\longrightarrow} \tau_r\]
the following typing:
\[this: C_{\{r\}}, \overline{x}:\overline{\sigma}\vdash e:\tau\ \&\ (\mathcal{U}, \mathcal{A}_r\cdot X_r\cup \mathcal{V}(\mathfrak{X}-\{X_r\}))\]
is derivable and there is an assignment $\eta : \mathfrak{X}\to\mathcal{M}_{\leq \omega}$
satisfying:
\[\eta(X_r) = \mathcal{A}_r^*\cdot \mathcal{V}(\eta) \cup \mathcal{A}_r^\omega\ \wedge\ \mathcal{V}_r\supseteq \eta(X_r)\]
That is, a program is well-typed if and only if the constraints produced by the region-based B\"uchi type and effect system are satisfiable with respect to region-based type and effect declarations for all classes, all regions and all methods defined in this program.

\end{document}